\documentclass[12pt]{iopart}

\expandafter\let\csname equation*\endcsname\relax
\expandafter\let\csname endequation*\endcsname\relax
\usepackage{amsmath}

\providecommand{\binom}[2]{\left(\begin{array}{c}{#1}\\{#2}\end{array}\right)}

\newcommand{\CC}{\ensuremath{\mathcal{C}} }
\newcommand{\MM}{\ensuremath{\mathcal{M}} }
\newcommand{\LL}{\ensuremath{\mathcal{L}} }
\newcommand{\None}{\ensuremath{\mathbb{N}}}
\newcommand{\Ztwo}{\ensuremath{\mathbb{Z}^2}}
\newcommand{\Ntwo}{\ensuremath{\mathbb{N}^2}}
\newcommand{\JJ}{\ensuremath{\mathcal{J}} }

\usepackage{graphicx}
\usepackage{iopams}

\usepackage{amsthm}
\usepackage{enumerate}

\usepackage[hidelinks]{hyperref}

\theoremstyle{plain}
\newtheorem{theorem}{Theorem}[section]
\newtheorem{prop}[theorem]{Proposition}
\newtheorem{lem}[theorem]{Lemma}

\theoremstyle{definition}
\newtheorem{defn}[theorem]{Definition}
\newtheorem{rem}[theorem]{Remark}

\usepackage[usenames, dvipsnames]{color}
\definecolor{cb-blue}{RGB}{0,109,219}
\newcommand{\corr}[1]{\textcolor{black}{#1}}

\begin{document}

\title[Transfer operator approach to ray-tracing
in circular domains]{Transfer operator approach
to ray-tracing in circular domains}
\author{J Slipantschuk$^1$, M Richter$^{2,3}$,
D J Chappell$^3$, G Tanner$^2$, W Just$^1$
and O F Bandtlow$^1$}
\address{$^1$ School of Mathematical Sciences, Queen Mary University
of London, London E1 4NS, UK}
\address{$^2$ School of Mathematical Sciences, University of Nottingham,
University Park, Nottingham NG7 2RD, UK}
\address{$^3$ School of Science and Technology, Nottingham Trent University,
Nottingham NG11 8NS, UK}
\eads{
\mailto{j.slipantschuk@qmul.ac.uk},
\mailto{martin.richter@nottingham.ac.uk},
\mailto{david.chappell@ntu.ac.uk},
\mailto{gregor.tanner@nottingham.ac.uk},
\mailto{w.just@qmul.ac.uk},
\mailto{o.bandtlow@qmul.ac.uk}
}

\date{3 October 2019}

\begin{abstract}
The computation of wave-energy distributions in the mid-to-high frequency regime
can be reduced to ray-tracing calculations.
Solving the ray-tracing problem in terms of an operator
equation for the energy density leads to an inhomogeneous equation which involves a Perron-Frobenius operator
defined on a suitable Sobolev space. Even for fairly simple geometries, let alone realistic scenarios such as typical boundary
value problems in room acoustics or for mechanical vibrations, numerical approximations are necessary. Here we study the convergence of approximation schemes
by rigorous methods. For circular billiards we prove that convergence of
finite-rank approximations using a Fourier basis follows a power law where the
power depends on the smoothness of the source distribution driving the system.
The relevance of our studies for more general geometries is illustrated by
numerical examples.
\end{abstract}



\maketitle

\section{Introduction}

Ray-tracing methods serve as an important toolkit in finding approximate
solutions of linear wave equations in the high frequency limit.
This approximation is used in a variety of fields providing, for example, the connection between Maxwell's
equations and geometric optics, as well as between quantum
mechanics and classical Hamiltonian mechanics~\cite{Haa2001}.
The ray-tracing limit has also been considered in detail in acoustics, seismology and mechanical vibrations~\cite{TS07}.
In engineering applications, ray tracing is employed in handling electromagnetic problems, such as coverage estimates for 5G or WiFi communication~\cite{Des1972}, room acoustics simulations~\cite{SS15} as well as structure-borne sound propagation in mechanical structures~\cite{ChaIh2001}.
Finding closed form, analytical solutions to such engineering problems of sufficient complexity
is generally impossible, even using ray-tracing techniques,
and one has to use numerical methods instead.

For solving linear wave problems such as those listed above, the numerical methods used have to be adapted to the relevant length and frequency
scales involved.
In the low frequency regime, finite element methods (FEM) are routinely employed for resolving the full wave dynamics.
However, the number of degrees of freedom in an FEM model needs to scale with the wavelength and there is thus an upper limit in
frequency above which the required computational resources become unfeasible.
At very high frequencies, power balance approaches can often be used
as long as certain assumptions on the ergodicity of the underlying ray dynamics are
 satisfied~\cite{Tann_JSV09}.
In the mid-to-high frequency range, ray-tracing becomes the method of choice; standard ray-tracing techniques track all possible rays from a source to a receiver point~\cite{SS15} --- a method which becomes cumbersome if many reflections need to be taken into account. As an alternative {\em Dynamical Energy
Analysis} (DEA) was proposed and has proven to be useful in particular for structure-borne sound problems~\cite{Tann_JSV09, HarMorTanCha2019_tractor_yanmar}.
Instead of tracking individual rays carrying vibrations across the
complex structure --- which is extremely challenging --- in DEA, the problem is reformulated in terms of densities of
rays, which are then mapped across a mesh representing the structure~\cite{ChLoSoTa_WM14, ChTaLoSo_PRSA13}.
This reduces the ray tracing problem from tracking rays on complicated and
curved domains to mapping ray segments across small, plane patches of a simple shape forming the mesh, typically
triangular or quadrilateral mesh cells.
The ray densities are then mapped from one cell of the mesh to
adjacent ones and the overall transport problem can be formulated in terms
of an inhomogeneous
equation of the form
\begin{equation}\label{eq:op_eq}
(I-\LL)f=f_0\,,
\end{equation}
where
$f_0$ is the initial ray density, $\LL$ a Perron-Frobenius type operator describing the
evolution of ray densities and $f$ the required final ray density.
%
%
Using DEA, the distribution of vibrational energy in
mechanical structures, such as ships, cars and tractors~\cite{HaTaXiChBa_JP16, HarMorTanCha2019_tractor_yanmar} can be calculated
successfully.

For such realistic geometries, equation (\ref{eq:op_eq}) above cannot be solved analytically, so recourse is made
to numerical schemes based on heuristic finite-dimensional matrix approximations of the operator $\LL$.
%
To date, very little is known about the convergence properties of these schemes and the dependence of the convergence rate on the
ray dynamics, as well as the discretisation techniques \cite{ChTaLoSo_PRSA13}.
The precise form of convergence is likely to be highly sensitive to both the basis functions used
in approximating the inhomogeneous equation (\ref{eq:op_eq}), as well as dynamical and damping properties of the system under
investigation~\cite{HaTaXiChBa_JP16}.
%
%
For our study, we will therefore be concerned with the approximation of
$\LL$ by operators of finite rank.
\corr{There is a plethora of papers on numerical approximation of Perron-Frobenius operators, starting with Ulam's method of phase space discretisation, finite section or Galerkin methods,
and data-driven methods, see for example
\cite{BaladiHolschneider1999, DelJunge99, DelFroSer00, Klus16,
  Liverani2001}
to mention but a few.}
Surprisingly,
the application of DEA (which falls into the Galerkin category)  to even fairly simple geometries has not been dealt with
at a rigorous level. Here, we shall thus focus
on one of the simplest cases, the
billiard dynamics given by the ballistic motion within
a circular disk. We shall establish rigorous error bounds of finite-dimensional
approximations for the resulting energy distribution.

In order to set up the required notation, consider
a particle moving inside a circular billiard table $\cal{D}$
being specularly reflected at its boundary $\partial \cal{D}$.
We parametrise $\partial \cal{D}$ by the
polar angle
$x\in\mathbb{R}/2\pi\mathbb{Z}$ and we denote by
$y\in [-\pi/2, \pi/2]$ the angle of
reflection that the postcollisional velocity vector
has with the inward normal to $\partial \cal{D}$.
Initially the collision angle is defined on an interval.
It is, however, technically simpler to deal with cyclic variables.
Since both angles $-\pi/2$ and $\pi/2$ correspond to a particle
which sticks on the boundary we identify both angles
so that the collision angle becomes a cyclic variable as well.
With these conventions, the collision map $T$ on the domain
$\Omega = (\mathbb{R}/2\pi\mathbb{Z}) \times (\mathbb{R}/\pi\mathbb{Z})$ can be written as
\begin{equation}\label{eq:T}
T(x,y) = (x + \pi -2y, y), \quad \, (x, y)\in \Omega
\end{equation}
with its inverse $\phi=T^{-1}$ given by
\begin{equation}\label{eq:phi}
\phi(x,y) = (x - \pi + 2y, y) \quad \, (x, y)\in \Omega\,.
\end{equation}
It is not difficult to see that the collision map $T$ preserves the normalised Haar measure on
$\Omega$. The long-term statistical behaviour of $T$ can thus be studied by
investigating the associated Perron-Frobenius operator (see, for example, \cite{BoyarskyGora}),
which for invertible measure-preserving
maps is given by the composition operator $\CC_\phi$ defined as
\begin{equation}\label{eq:Cphi}
(\CC_\phi f)(x, y) = f(\phi(x, y)), \quad (x,y) \in \Omega\,,
\end{equation}
where $f\colon \Omega \to \mathbb{C}$.
In the current work we are interested in the properties of a weighted Perron-Frobenius operator,
also known as a transfer operator.
In order to define it, let us first introduce
a multiplication operator $\MM_w$ acting on functions $f\colon \Omega \to \mathbb{C}$ by
\begin{equation}\label{eq:Mw}
(\MM_w f)(x, y) = w(x,y) f(x,y), \quad (x,y) \in \Omega,
\end{equation}
where $w\colon \Omega\to [0,\infty)$ is a suitable weight function, which in
the DEA framework accounts for
dissipation caused either by collisions with the wall or by in-flight
dissipation.
The transfer operator, understood to be acting on a suitable space of functions detailed in the
following section, is now given by
\begin{equation}\label{eq:op_eq_b}
\LL = \MM_w \CC_\phi\,.
\end{equation}
In the present article, we are interested in approximations of the solution to the
operator equation (\ref{eq:op_eq})
with $f_0\colon \Omega \to [0,\infty)$  interpreted as the initial boundary density of particles induced
by the first boundary collision of particles emitted by a source located in the interior
of $\cal{D}$ (see \cite{Tann_JSV09}).
In the DEA approach this quantity represents the energy source.
The resulting energy distribution is captured by the solution,
$f\colon \Omega \to [0,\infty)$, which gives the stationary boundary density
generated by the collision dynamics.
Given a suitable Banach space and a sequence of finite-rank projections
$(P_K)_{K\in\None}$, an approximation method for \eqref{eq:op_eq} can be
constructed by considering the projected finite-dimensional problem
\begin{equation}\label{eq:proj_op_eq}
\corr{(I-P_K\LL P_K)f_K = f_0 \, .}
\end{equation}

The aim of this work is to present a Banach space for $f_0$ and
$(P_K)_{K\in\None}$, so that problem \eref{eq:proj_op_eq}
has solutions, which converge in a suitable topology to the solution
of \eref{eq:op_eq} as $K$ tends to infinity, with the speed of convergence
being of the order $K^{-\alpha}$. The exponent $\alpha$
depends on the smoothness of $f_0$
and the requirements imposed on the type of convergence.

In passing we note that transfer operators have their roots in
statistical mechanics \cite{ Mayer78, Ruelle78}
and nowadays play an important role in the ergodic theory of smooth expanding, or more
generally, hyperbolic dynamical systems (see, for example, \cite{BaladiBook1, BaladiBook2}).
The main reason for their popularity in this context derives from the fact that for expanding or
hyperbolic dynamical systems the transfer operator, when considered on a
suitable function space, can be shown to have discrete peripheral spectrum, from which
long-term statistical properties of the underlying system can be derived. In the elliptic setting,
however, such as for the circular billiard considered in this article, analogous results cannot be
expected, and, as a consequence, transfer operator methods have received little attention in this
context. It is perhaps worth noting that in our setting we do not require discreteness of the
peripheral spectrum of the transfer operator. The main onus is to show that the resolvent of the
transfer operator exists at the point $1$ (see equation \eqref{eq:op_eq})
and can be effectively
approximated by finite-rank operators (see equation \eqref{eq:proj_op_eq}).

As we intend to keep our presentation accessible to non-specialists,
we will occasionally elaborate on aspects covered
in the specialised literature but which may not be well known
to a general audience. The remaining parts are organised as follows.
In Section~\ref{sec:trans} we introduce Sobolev spaces, on which the transfer
operator and its finite-dimensional approximations are bounded operators
with spectral radii bounded away from $1$.
In Section~\ref{sec:conv} we shall prove the convergence results for
the operator equations \eref{eq:op_eq} and \eqref{eq:proj_op_eq} stated as Theorem \ref{thm:conv}.
In the final Section~\ref{sec:conc}  we summarise the main findings,
compare the formal results with numerical simulations
and explore the relevance of the current study in a wider context.

%
%

\section{Sobolev spaces and transfer operators}\label{sec:trans}

We will be interested in certain subspaces of
$L^2(\Omega) = L^2(\Omega, m)$ where
$dm = dx dy/(2\pi^2)$ is the normalised two-dimensional Lebesgue measure on $\Omega$.
The natural inner product is given by
\[(f,g)_{L^2} = \int_{\Omega} f(x,y) \overline{g(x,y)} {\rm d}m \, . \]
An orthonormal basis of $L^2(\Omega)$ is given by
$\{e_k\colon k\in \Ztwo\}$ where  $e_k(x,y)=\mathrm{e}^{\mathrm{i} k_1 x}
\mathrm{e}^{2\mathrm{i}k_2 y}$
so that $f(x,y) = \sum_{k \in \Ztwo} c_k(f) e_k(x,y)$ with Fourier coefficients
$c_k(f) = (f, e_k)_{L_2}$.

\begin{defn}
Let $m=(m_1, m_2) \in\mathbb{N}_0^2$. The Sobolev space
$H^m(\Omega)$ is the collection of all $f\in L^2(\Omega)$ such that for all
$\nu = (\nu_1, \nu_2) \in \mathbb{N}^2$ with $\nu_1 \leq m_1$  and $\nu_2 \leq m_2$ the
weak derivatives $D^\nu f = D_x^{\nu_1}D_y^{\nu_2} f$  exist and belong to $L^2(\Omega)$.
\end{defn}

The space $H^m(\Omega)$ is a Hilbert space, when equipped with the inner product\footnote{This choice of inner product is sometimes referred to as the modified inner product, in contrast with the classical one (see, for example, \cite[Def 2.2]{KSU}).}
\begin{equation}\label{eq:Hm}
(f,g)_{H^m} =
(f,g)_{L^2} +(D^{m_1}_xf, D^{m_1}_xg)_{L^2} +
(D^{m_2}_yf,D^{m_2}_yg)_{L^2} \, .
\end{equation}
One can rewrite this definition in terms of Fourier coefficients.
Using the fact that $c_k(D^\nu f) = (\mathrm{i}k_1)^{\nu_1}
(2\mathrm{i}k_2)^{\nu_2} c_k(f)$, equation \eqref{eq:Hm} can be expressed as
\begin{equation}\label{eq:Hkf}
(f,g)_{H^m} = \sum_{k\in\Ztwo}(1+|k_1|^{2m_1} + |2k_2|^{2m_2})
c_k(f) \overline{c_k(g)} \, .
\end{equation}

\begin{rem}
For $m=(m_1,m_2)$ with $m_1=m_2$ the Sobolev space $H^m(\Omega)$
coincides with the
classical isotropic Sobolev space, while for $m_1\neq m_2$, the space
is an example of an anisotropic Sobolev space (see, for example, \cite[Sec. 2.2]{ChenWang}).
\end{rem}

Using equation \eqref{eq:Hkf} we can define fractional Sobolev spaces
$H^s(\Omega)$  for  $s=(s_1, s_2) \in \mathbb{R}^2_+$ as
\[
H^{s}(\Omega) = \left\{f \in L^2(\Omega)\colon
\sum_{k\in \Ztwo} |c_{k}(f)|^2 (1+|k_1|^{2s_1} + |2k_2|^{2s_2})< \infty \right\},
\]
which are Hilbert spaces when equipped with the inner product given in
equation \eqref{eq:Hkf} with
$m$ replaced by $s$.

We shall next investigate the properties of the composition operator
$\CC_\phi$ associated with the map $\phi$ in \eqref{eq:phi} on the fractional
Sobolev space $H^s(\Omega)$.

\begin{lem} \label{lem:C}
The composition operator $\CC_{\phi}$ given in \eqref{eq:Cphi} considered
on $H^s(\Omega)$  with $s_1 \geq s_2 \geq 0$ is bounded, with spectral radius $r(\CC_{\phi})=1$.
\end{lem}

\begin{proof}
For any $n\in \mathbb{N}$ and $(x,y)\in\Omega$
we have $\phi^n(x,y) = (x - n\pi + 2ny, y)$, and thus
\begin{equation}
  \label{eq:Cphi_n}
  (\CC^n_{\phi} e_k)(x,y) =
  (\CC_{\phi^{n}}e_k)(x,y) = (-1)^{k_1n} e_{k_1, nk_1 + k_2}(x,y),
\end{equation}
for any $k\in\Ztwo$.

In order to show that the operator is bounded
we will need the following general inequality.
Let $(x,y)\in[0,\infty)^2$ and $t\geq 0$, then
\begin{equation}\label{eq:holder}
(x + y)^t \leq C_t (x^t + y^t), \quad \mbox{ with } C_t = \max(1, 2^{t-1}).
\end{equation}

 Using equation \eqref{eq:holder} we obtain the bound
 $|nk_1+k_2|^{2s_2} \leq
 C_{2s_2} (n^{2s_2} |k_1|^{2s_1} + |k_2|^{2s_2})$ for $s_1\geq s_2$, which leads to
  \[\|\CC^n_{\phi} e_k\|_{H^s}^2 =
 1 + |k_1|^{2s_1} + |2(nk_1+k_2)|^{2s_2} \leq
 \left(1+C_{2s_2}(2n)^{2s_2}\right) \|e_k\|^2_{H^s}.\]
Since $(\CC^n_\phi e_k, \CC^n_\phi e_l)_{H^s}=0$ for $k\neq l$,
the operator norm of $\CC^n_{\phi}$ is bounded from above by
$\left(1+(2n)^{2s_2}\max(1, 2^{2s_2-1})\right)^{1/2}$, resulting in
the upper bound for the spectral radius
 \[r(\CC_{\phi}) =
 \lim_{n\to\infty} \|\CC^n_{\phi}\|^{1/n}
 \leq \lim_{n\to\infty}\left(1+(2n)^{2s_2}\max(1, 2^{2s_2-1})\right)^{1/(2n)}=1
\, .
\]
In order to see that the inequality above is an equality, observe that the operator norm of
$\CC^n_{\phi}$ is bounded from below by $1$ as
 $\|\CC^n_{\phi}e_0\|_{H^s}=\|e_0\|_{H^2}$. Thus $r(\CC_{\phi})=1$.
 \qedhere
\end{proof}

Before proceeding we note that by \eqref{eq:Cphi_n}, the action of the
composition operator on $H^s(\Omega)$ can be represented
by the action of the matrix
\[ A =
\begin{pmatrix}
1  & 0 \\ 1 & 1
\end{pmatrix} \]
on Fourier coefficients.
In particular, we have
\begin{equation}\label{eq:TMatCirc}
\CC^{n}_\phi e_{k} = (-1)^{k_1n} e_{A^nk}.
\end{equation}

For $K \in \mathbb{N}$ define
$\Lambda_K = \Lambda_K^0 = \{(k_1, k_2)\in\Ztwo\colon |k_1|<K, |k_2|<K\}$, and let
$\Lambda_K^n = A^n(\Lambda_K)$. Then for any $n\in\mathbb{N}_0$ we can define a finite-rank operator  $P_{\Lambda_K^n}\colon
H^s(\Omega) \to H^s(\Omega)$ by

\begin{equation}\label{eq:PK}
(P_{\Lambda^n_K}f)(x,y) = \sum_{k\in\Lambda^n_K} c_k(f) e_k(x,y), \quad (x,y)\in \Omega\,.
\end{equation}

\begin{lem}\label{lem:CphiPK}
Let $\CC_\phi$ and $P_{\Lambda_K}$ be as above. Then
\[\CC^n_\phi P_{\Lambda_K} = P_{\Lambda^n_K} \CC^n_\phi\]
for any $n, K \in \None_0$.
\end{lem}

\begin{proof}
This follows by checking the equality for any basis element $e_k$ and
noting that $A^n$ is invertible.
\end{proof}

\begin{defn}\label{def:damping}
Let $\MM_w$ denote the multiplication operator as defined in equation
\eqref{eq:Mw},
considered as an operator on $H^s(\Omega)$, with a smooth weight function
$w\colon \Omega \to [0,\infty)$. In addition, we assume that $w$ has the following properties:
\begin{enumerate}[(a)]
 \item $\|w\|_\infty = \sup_{x\in \Omega} |w(x)| < 1$;
 \item $w$ is bounded away from zero;
 \item $w(x,y)=w(x',y)$ for any $(x,y), (x',y)\in \Omega$, that is, the weight
 $w$ does not depend on the first argument.
\end{enumerate}
\end{defn}

\begin{rem} The operator $\MM_w$ models the effect of damping on the
  motion of the billiard particle. Assumptions (a) and (b) imply that
  the damping is well-behaved, while assumption (c) is innocuous,
  given the circular symmetry of the billiard table.
\end{rem}

The following two lemmas summarise basic properties of $\MM_w$ and $\CC_\phi$.
\begin{lem}\label{lem:identities} Let $\MM_w, \CC_\phi$ and $P_{\Lambda_K}$
be as above. Then we have the following.
\begin{enumerate}[(i)]
\item $\MM_w\CC_\phi = \CC_\phi \MM_w$;
\item $D_{x}\CC_\phi = \CC_\phi D_{x}$;
\item $D_{x}\MM_w = \MM_wD_x$;
\item $D_{y}\CC^n_\phi = 2n\CC^n_\phi D_x + \CC^n_\phi D_y$ for
$n\in \None$;
\item $D_y \MM_w^n = n \MM_w^{n-1}\MM_{D_yw} + \MM_w^nD_y$ for
$n\in \None$;
\item $D_x P_{\Lambda_K}=P_{\Lambda_K} D_x$ and
$D_y P_{\Lambda_{K}}=P_{\Lambda_K} D_y$ for $K\in\None$.
\end{enumerate}
\end{lem}
\begin{proof}
Items {\it (i)} and {\it (iii)} follow from Definition \ref{def:damping}(c);
items {\it (ii)} and {\it (iv)} follow by direct computation using the map $\phi$;
item {\it (v)} is obvious and {\it (vi)} is a direct consequence
of the relations
$c_k(D_xf) = (\mathrm{i}k_1)c_k(f)$ and $c_k(D_yf) = (2\mathrm{i}k_2)c_k(f)$.
\end{proof}

We write $\LL_K = P_{\Lambda_K} \MM_w \CC_\phi P_{\Lambda_K}$ for the finite-rank approximation of $\LL=\MM_w\CC_\phi$.
Using Lemma \ref{lem:identities} {\it (i)}
and Lemma \ref{lem:CphiPK}, we can write
$\LL_K^n$ for $n\in\None$ as
\begin{equation}\label{eq:LK_decomp}
\LL^n_K = (P_{\Lambda_K} \MM_w \CC_\phi P_{\Lambda_K})^n = P_{\Lambda_K}
\left(\prod_{l=1}^n \MM_w P_{\Lambda^l_K}\right) \CC^n_{\phi}.
\end{equation}

In order to state the properties of $\LL$ and $\LL_K$ we need to introduce the
following multi-index notation:
an $n$-dimensional multi-index is an $n$-tuple
$\boldsymbol{i}_n = (i_1, i_2, \ldots, i_n)$
of non-negative integers of order
$|\boldsymbol{i_n}| = i_1 + i_2 + \cdots + i_n = m$; the corresponding
multinomial coefficient is given by
\[ \binom{m}{\boldsymbol{i}_n} =
\frac{m!}{i_1!i_2!\cdots i_n!}\, .\]

\begin{lem}\label{lem:identities2}
Let $\MM_w, \CC_\phi$ and $P_{\Lambda^l_K}$ be as above. Then
we have the following.
\begin{enumerate}[(i)]
\item $D_y^m \CC_\phi = \sum_{i=0}^m 2^{m-i}\binom{m}{i}
  \CC_\phi D_x^{m-i} D_y^i$;
\item $D_y^m \CC_\phi^n =
  \sum_{|\boldsymbol{i}_{n+1}|=m}
  2^{m-i_{n+1}} \binom{m}{\boldsymbol{i}_{n+1}}
  \CC_\phi^n D_x^{m-i_{n+1}}D_y^{i_{n+1}}$;

\item $D_y^m \left(\prod_{l=1}^{n} \MM_w P_{\Lambda^l_K}\right) =
\sum_{|\boldsymbol{i}_{n+1}|=m}
\binom{m}{\boldsymbol{i}_{n+1}}
\left(\prod_{l=1}^{n} \MM_{D^{i_l}_yw} P_{\Lambda^l_K}\right) D_y^{i_{n+1}}$.
\end{enumerate}
\end{lem}

\begin{proof}
Item {\it (i)} follows by induction over $m$ using Lemma \ref{lem:identities}{\it (iv)} for the base case
$m=1$. For item {\it (ii)}, the additional induction over $n$ follows by
rewriting {\it (i)} as $D_y^m \CC_\phi = \sum_{i_1 + i_2 = m}2^{i_1} \binom{m}{i_1, i_2} \CC_\phi
D^{i_1}_xD_y^{i_2}$.
Finally, item {\it (iii)} follows from the Leibniz formula.
\end{proof}

We are now ready to prove the main result of this section. Keeping in mind that we assume that
the billiard dynamics is dissipative, that is, the weight is chosen so that $\|w\|_\infty<1$, the
following lemma shows that, given $f_0\in H^s(\Omega)$, the problem \eref{eq:op_eq}
and the projected version \eref{eq:proj_op_eq} have unique
solutions $f\in H^s(\Omega)$ and $f_K\in H^s(\Omega)$, respectively.

\begin{lem}\label{lem:LK_norms}
Consider $\LL$ and $\LL_K$, $K\in \None$, as operators on
$H^{s}(\Omega)$ for $s\in\mathbb{N}_0^2$ with $s_1 \geq s_2 \geq 0$.
Then
\begin{enumerate}[(i)]
\item $(\LL_K)_{K\in \mathbb{N}}$ is a family of bounded operators $H^{s}(\Omega)$ with
norms bounded uniformly in $K$. Moreover, $r(\LL_K) \leq \|w\|_{\infty}$
for all $K$;
\item $\LL$ is a bounded operator on $H^s(\Omega)$ with $r(\LL) \leq \|w\|_{\infty}$.
\end{enumerate}
\end{lem}

\begin{proof}
We shall only prove statement \textit{(i)}, as the proof of statement
\textit{(ii)} follows by almost identical arguments.
In the following, we shall assume that
$s_1\geq s_2 \geq 1$, as the case $s_1s_2=0$ follows by identical arguments.
For $f\in H^s(\Omega)$ we have
\begin{equation}\label{eq:normLK_s}
\|\LL^n_Kf\|^2_{H^s} =
\|\LL^n_Kf\|^2_{L^2} + |
|D^{s_1}_x\LL^n_Kf\|^2_{L^2} + \|D^{s_2}_y\LL^n_Kf\|^2_{L^2}\,.
\end{equation}

Let $p,q\in\None$ with $p \leq s_1$ and $q \leq s_2$.
It is not difficult to see that for any $f\in H^s(\Omega)$ and
$K \in \None_0$ the following holds.
\begin{enumerate}[(a)]
\item $\|P_{\Lambda_K^j} f\|_{L^2} \leq \|f\|_{L^2}$ for any $j\in \None_0$;
\item $\|\MM_wf\|_{L^2}  \leq \|w\|_{\infty} \|f\|_{L^2}$;
\item $\|D^p_xf\|^2_{L^2} \leq \|D^{s_1}_xf\|^2_{L^2}$ and
$\|D^q_xf\|^2_{L^2} \leq \|D^{s_2}_xf\|^2_{L^2}$;
\item $\|D^p_xD^q_yf\|^2_{L^2} \leq
\|D^{p+q}_xf\|^2_{L^2} + \|D^{p+q}_yf\|^2_{L^2}$
wherever $p+q\leq s_2$.
\end{enumerate}
Here, statements (c) and (d) follow by writing the $L^2$ norm of $D_x^pf$ and
$D_y^qf$ using Parseval's identity.

Writing $\LL^n_K$ as in equation \eqref{eq:LK_decomp} and using (a) and (b) above iteratively we have
\begin{equation}\label{eq:const_term}
\|\LL^n_Kf\|_{L^2}  =
\|(P_{\Lambda_K} \MM_w \CC_\phi P_{\Lambda_K})^nf\|_{L^2} \leq
\|w\|^n_{\infty} \|\CC^n_\phi f\|_{L^2}
\leq \|w\|^n_{\infty} \|f\|_{L^2},
\end{equation}
where the last inequality follows from the fact that the
operator norm of $\CC_\phi$ on $L^2(\Omega)$ equals $1$.

As
$D_x$ commutes with any of the operators involved
(Lemma \ref{lem:identities} \textit{(ii,iii,vi)}) we have
in the second term on the right-hand side of \eqref{eq:normLK_s}
that $D^{s_1}_x\LL_K^n =
\LL_K^nD^{s_1}_x$. By the same argument as above we have
\begin{equation}\label{eq:Dx_term}
\|D^{s_1}_x\LL^n_Kf\|_{L^2} \leq \|w\|^n_{\infty} \|D^{s_1}_xf\|_{L^2}.
\end{equation}
In order to bound the last term in equation \eqref{eq:normLK_s} we are
using  Lemma \ref{lem:identities2}\textit{(iii)} and H\"older's inequality in order to write
\[
\|D_y^{s_2} \LL_K^nf \|_{L^2}^2
= \|\sum_{j=0}^{s_2} A_j D^j_y\CC_\phi^n f\|^2_{L^2}
\leq (s_2+1) \sum_{j=0}^{s_2}\|A_j\|^2_{L^2}
\| D^j_y\CC_\phi^n f\|^2_{L^2},
\]
where
$A_j = \sum_{|\boldsymbol{i}_n|=s_2-j}
\binom{s_2}{\boldsymbol{i}_n,j}
\left(\prod_{l=1}^{n}  M_{D^{i_l}_yw} P_{\Lambda^l_K}\right)$.

We shall first obtain a bound for $\|D^j_y\CC_\phi^n f\|_{L^2}$.
Using Lemma \ref{lem:identities2}\textit{(ii)} and
decomposing the sum in terms of powers of $D_x$ and $D_y$
we obtain
\[D^j_y\CC_\phi^n   =
 (2n)^{j} C_{\phi}^n D_x^j + C_{\phi}^n D_y^j +
\sum_{|\boldsymbol{i}_{n+1}|=j \atop 0 < i_{n+1} < j}
2^{j-i_{n+1}} \binom{j}{\boldsymbol{i}_{n+1}}
\CC_\phi^n D_x^{j-i_{n+1}}D_y^{i_{n+1}}, \]
where we have used the multinomial formula $\sum_{|\boldsymbol{i}_{n}|=k}
\binom{k}{\boldsymbol{i}_{n}}=n^k$.
Thus, for $j\leq m$ we obtain
using H\"older's inequality, the multinomial formula and
upper bounds for $2^{j-i_{n+1}}$ 
\begin{eqnarray}\label{eq:DjyCphi_norm}
\|D^j_y\CC_\phi^n f\|^2_{L^2}
&\leq& 2^j(n+1)^j \left(
2^jn^j  \|D_x^jf\|^2_{L^2} + \|D_y^jf\|^2_{L^2} \right) \nonumber \\
& &+ 2^j(n+1)^j \left((2^j(n+1)^{j}-2^jn^j-1) \max_{0<i<j}
\|D_x^{j-i}D_y^{i}f\|^2_{L^2}\right)  \nonumber \\
& \leq& 2^{2s_2}(n+1)^{2s_2} \left(\|D_x^{s_1}f\|^2_{L^2} + \|D_y^{s_2}f\|^2_{L^2}\right),
\end{eqnarray}
where the last inequality uses (c) and (d).

Next we shall obtain a bound on the operator norm of $A_j$ for $j\leq s_2$.
First note that $\MM_{D_y^lw} = \MM_w \MM_{(D_y^lw)/w}$ is
well-defined as $w$ is bounded away from zero.
By using (a) and (b) iteratively, for any
$\boldsymbol{i}_n =(i_1, \ldots, i_n)$ with $|\boldsymbol{i}_n| = s_2$ we have
\[
\|\prod_{l=1}^{n}  \MM_{D^{i_l}_yw} P_{\Lambda^l_K}f\|^2_{L^2}
\leq C_{s_2} \|w\|^{2n}_\infty \|f\|^2_{L^2}
\]
where $
C_{s_2} = \max_{i_1,\ldots, i_n \atop |\boldsymbol{i}_{n}|=s_2}
\left(\prod_{l=1}^{n}
\|D^{i_l}_yw/w\|^2_\infty\right) \leq
\max_{0\leq l \leq s_2} \|D^{l}_yw/w\|_\infty^{2s_2}$
is a constant independent of $n$.
Using arguments analogous to those used to obtain inequality \eref{eq:DjyCphi_norm},
we obtain the bound
\begin{equation}\label{eq:Ajnorm}
\|A_j\|_{L^2}^2 \leq (n+1)^{2s_2} C_{s_2} \|w\|_\infty^{2n}\,. 
\end{equation}

Using the estimates \eref{eq:const_term},
\eref{eq:Dx_term}, \eref{eq:DjyCphi_norm} and \eref{eq:Ajnorm} in
equation \eqref{eq:normLK_s} we arrive at the bound
\[\|\LL^n_Kf\|^2_{H^s} \leq \tilde{C}_{n,s_2} \|w\|_\infty^{2n} \|f\|^2_{H^2}\,\]
with $\tilde{C}_{n,s_2} \leq (s_2+1)s_2(n+1)^{4s_2}2^{2s_2}C_{s_2} + 1$.
As $\tilde{C}_{n,s_2}$ is independent of $K$, the family
$(\LL_K)_{K\in \mathbb{N}}$ is a uniformly bounded family of bounded operators on $H^s(\Omega)$.
Finally, taking the right hand side of
equation \eqref{eq:normLK_s} to the power of $1/n$ and observing that
$\tilde{C}_{n,s_2}$ grows polynomially in $n$,  the upper bound for
the spectral radius of $\LL_K$ follows.
\end{proof}

\section{Convergence properties}\label{sec:conv}
In the previous section we established (see Lemma~\ref{lem:LK_norms}) that
given $f_0\in H^s(\Omega)$, the problem \eref{eq:op_eq}
and the projected version \eref{eq:proj_op_eq} have unique
solutions $f\in H^s(\Omega)$ and $f_K\in H^s(\Omega)$, respectively. We shall now turn to
establishing  the convergence of $f_K$ to $f$. This would be straightforward if we knew that
$\LL_K\to \LL$ as $K\to \infty$ in the operator norm on $H^s(\Omega)$, since then, using the
so-called second resolvent identity
\begin{equation}
\label{eq:sec_res_id}
(I-\LL)^{-1}-(I-\LL_K)^{-1}=
-(I-\LL)^{-1}(\LL-\LL_K)(I-\LL_K)^{-1}\,,
\end{equation}
we would have
\[ \|f-f_K\|_{H^s}=\|(I-\LL)^{-1}f_0-(I-\LL_K)^{-1}f_0 \|_{H^s}=
\|(I-\LL)^{-1}(\LL_K-\LL)(I-\LL_K)^{-1}f_0 \|_{H^s}\,,
\]
from which convergence of $f_K\to f$ in $H^s(\Omega)$ could be readily obtained.

This, however,
cannot be the case, as if $\LL_K\to \LL$ as $K\to \infty$ in the operator norm on $H^s(\Omega)$,
then $\LL$, as a uniform limit of finite-rank operators, would be compact on $H^s(\Omega)$.
However, as $\LL$ has a bounded inverse on $H^s(\Omega)$, it cannot be compact.

We thus need to resort to a slightly weaker notion of convergence, that is, we shall consider the
transfer operator as an operator between Sobolev spaces of different order. In passing, we
remark that this idea is also at the heart of one of the most successful techniques to obtain
spectral approximation results of transfer operators, where perturbation sizes are measured in
`triple' norms (see, for example, \cite{KL}).

In the following we shall explain this idea in more detail. We start with the following important
observation. For $t, s \in [0,\infty)^2$ with
$s_1 \geq s_2 > t_1 \geq t_2 \geq 0$,
functions in $H^s(\Omega)$ can be identified with functions in $H^t(\Omega)$ using the
embedding operator $\JJ\colon H^s(\Omega) \hookrightarrow H^t(\Omega)$ given by
$\JJ f = f$. This operator is not just continuous, but also compact, as the following
lemma shows.

\begin{lem}\label{lem:J}
Let $\JJ\colon H^s(\Omega) \hookrightarrow H^t(\Omega)$
be the canonical embedding, where $t, s \in [0,\infty)^2$ with
$s_1 \geq s_2 > t_1 \geq t_2 \geq 0$. Let $P_K = P_{\Lambda_K}$ the projection operator
in equation \eqref{eq:PK}, and $\JJ_K=\JJ P_{K}$. Then,
\[ \|\JJ - \JJ_K\|_{H^s\to H^t} \leq C (1+K^2)^{-\alpha/2} \]
for some $C>0$ and $\alpha=(s_2 - t_1)$.
\end{lem}

\begin{proof}
Let $f\in H^s(\Omega)$. Using the notation
$a_t(k) =  1+ |k_1|^{2t_1} + |2k_2|^{2t_2}$ we have
\[\|\JJ f - \JJ_Kf\|^2_{H^t} =
\sum_{i=1}^3 \sum_{k \in I_i(K)} |c_k(f)|^2 a_t(k),\]
with $I_1(K) = \{k\in \Ztwo\colon |k_1|\geq K, |k_2|\geq K\}$,
$I_2(K) = \{k\in \Ztwo\colon |k_1|< K, |k_2|\geq K\}$,
$I_3(K) = \{k\in \Ztwo\colon |k_1|\geq K, |k_2|< K\}$.
We will first show that there exists a constant $C'$ such that
\[a_t(k) \leq C' (1+|k_1|^2 + |2k_2|^2)^{-\alpha} a_s(k) \, . \]
For this, first observe that
\[(1+ |k_1|^{2} + |2k_2|^{2})^{s_2} \leq C_{s_2}
(1+ |k_1|^{2s_1} + |2s_2|^{2s_2}) \leq  C_{s_2} a_s(k), \]
which follows by H\"older's inequality and $s_1 \geq s_2$. Then,
\begin{align*}
a_t(k) &= 1+ |k_1|^{2t_1} + |2k_2|^{2t_1} \leq 3(1+ |k_1|^{2} + |2k_2|^{2})^{t_1} \\
 & = 3(1+ |k_1|^{2} + |2k_2|^{2})^{t_1-s_2} (1+ |k_1|^{2} + |2k_2|^{2})^{s_2} \\
 & \leq 3 C_{s_2} (1+ |k_1|^{2} + |2k_2|^{2})^{t_1-s_2} a_s(k).
\end{align*}

Now, by bounding from above
each $ (1+|k_1|^2 +|2k_2|^2)^{-\alpha}$ with its maximal value in each
of the sums, we obtain
\begin{align*}
\|\JJ f - \JJ_Kf\|^2_{H^t}
&\leq C'\left((1+5K^2)^{-\alpha} +
(1+4K^2)^{-\alpha} +
(1+K^2)^{-\alpha}\right) \|f\|^2_{H^s} \\
&\leq 3C'(1+K^2)^{-\alpha} \|f\|^2_{H^s}.\qedhere
\end{align*}
\end{proof}


We are now able to show that $\LL$ can be approximated by finite-rank operators
when considered as operators from $H^s$ to $H^t$.
\begin{prop}\label{prop:JLk}
Let $\LL_K = P_K \LL P_K$ be the finite-rank approximation of $\LL$
on $H^s(\Omega)$ with $s\in\None^2$ and $s_1 \geq s_2$.
Let $\JJ$ be as above and $t\in\None_0^2$ with $s_2>t_1\geq t_2$.
Then
\[\|\JJ(\LL_K - \LL)\|_{H^s\to H^t}
\leq C  (1+K^2)^{-\alpha/2}  \]
for some $C>0$ and $\alpha = s_2-t_1$.
\end{prop}

\begin{proof}
Let $\LL'$ denote the transfer operator when considered
on the larger space $H^t(\Omega)$. Then using the property $\JJ\LL = \LL'\JJ$,
we have
\[\JJ(\LL_k - \LL) = \JJ P_K\LL P_K - \JJ \LL
= (\JJ P_K - \JJ) \LL P_K - \LL' (\JJ P_K - \JJ).  \]
Thus,
\begin{align*}
\|\JJ(\LL_K - \LL)\|_{H^s\to H^t} &\leq
\left(\|\LL\|_{H^s\to H^s} \|P_K\|_{H^s\to H^s}
+ \|\LL'\|_{H^t\to H^t}\right)
\|\JJ - \JJ_K\|_{H^s\to H^t} \\
&\leq C (1+K^2)^{-\alpha/2},
\end{align*}
where we have used Lemma \ref{lem:LK_norms}, Lemma \ref{lem:J}
and $\|P_K\|_{H^s\to H^s} \leq 1$.
\end{proof}

\begin{prop}\label{prop:resolv}
Let $\LL$ and the family $(\LL_K)_K$ be as above, considered as operators
on $H^s(\Omega)$ where $s\in\Ntwo$ with $s_1 \geq s_2$.
Then, for $t\in\None_0^2$ with $s_2>t_1\geq t_2$ and for all $K\in\None$ we have
\[ \| (I-\LL_K) ^ {-1}  - (I - \LL)^{-1}\|_{H^s\to H^t} \leq C(1+K^2)^{-\alpha/2},\]
for some $C>0$ and $\alpha = s_2-t_1$.
\end{prop}

\begin{proof}
As $r(\LL) \leq \|w\|_{\infty} <1$ by Lemma \ref{lem:LK_norms},
the operator $(I-\LL)^{-1}$ exists and is bounded.
Let $(\LL'_K)_{K}$ denote the family of transfer operators when
considered on the larger space $H^t(\Omega)$. Similarly, as $\rho(\LL'_K)\leq \|w\|_{\infty} <1$ and the norms of $(\LL'_K)^n$ are bounded uniformly in $K$ by
Lemma \ref{lem:LK_norms}, the sums
$\sum_{n=0}^{\infty} \|\LL_K\|^n_{H^t\to H^t}$ are bounded by a constant
independent of $K$ and therefore
$\|(I-\LL_K)^{-1}\|_{H^t\to H^t}$ is uniformly bounded in $K$.

Using the property
$\JJ(I-\LL_K) = (I-\LL'_K)\JJ$ and the second resolvent identity
(see equation \eqref{eq:sec_res_id}) we have
\begin{align*}
\|\JJ ((I-\LL_K)^{-1}  &- (I - \LL)^{-1})\|_{H^s\to H^t} \\
&=\| (I-\LL'_K) ^ {-1} \JJ(\LL - \LL_K)(I - \LL)^{-1}\|_{H^t} \\
& \leq \|(I-\LL'_K) ^ {-1}\|_{H^t\to H^t}
\|\JJ(\LL - \LL_K)\|_{H^s\to H^t}
\|(I - \LL)^{-1}\|_{H^s\to H^s}\,.
\end{align*}
Using Proposition \ref{prop:JLk} for the bound on
$\|\JJ(\LL - \LL_K)\|_{H^s\to H^t}$ finishes the proof.
\end{proof}
We are finally able to state and prove our main convergence result.
\begin{theorem}\label{thm:conv}
Let $\LL$ and the family $(\LL_K)_K$ be as above, considered as operators
on $H^s(\Omega)$ with $s_1,s_2\in\None$ and $s_1 \geq s_2 > t_1 \geq t_2 \geq 0$.
Then for $f_0\in H^s(\Omega)$ the operator equations \eref{eq:op_eq}
and \eref{eq:proj_op_eq} have unique solutions $f\in H^s(\Omega)$
and $f_K\in H^s(\Omega)$, respectively. Moreover there exist a constant
$C>0$  such that for all $K\in\None$ we have

\[\|f - f_K\|_{H^t} \leq C (1+K^2)^{-\alpha/2} \|f_0\|_{H^s},\]
where $\alpha = s_2 - t_1$.
\end{theorem}
\begin{proof}
The statement follows by writing $f=(I-\LL)^{-1}f_0$,
$f_K=(I-\LL_K)^{-1}f_0$ and using Proposition \ref{prop:resolv}.
\end{proof}

\begin{rem}
\corr{Note that for $f_0 \in \Lambda_K = P_K(H^s(\Omega))$, the unique
solution $f_K$ to \eqref{eq:proj_op_eq} also lies in the finite-dimensional space $\Lambda_K$, so that \eqref{eq:proj_op_eq} can be solved as a truly finite-dimensional problem.}
\end{rem}


\section{Discussion and numerical experiments}\label{sec:conc}
Let us first summarise and rephrase our results in intuitive terms.
Since the linear
operator in equation \eqref{eq:op_eq} fails to be compact, any finite-dimensional matrix representation
would not reflect properties of the operator at all. Nevertheless the
finite-dimensional representation in \eqref{eq:proj_op_eq} provides a meaningful
approximation for the solution of the inhomogeneous equation.
For smooth periodic functions in
location and angle of reflection, the solution of the approximated
problem \eref{eq:proj_op_eq}
converges to the solution of \eref{eq:op_eq} in the Sobolev norm.
The approximation error depends on the degree of smoothness of
the inhomogeneous part. In addition, the approximation error is measured in
a weaker norm, for instance the frequently used $L^2$ norm for the choice
$t=(0,0)$. The properties of this weaker norm also determine the speed
of convergence. Broadly speaking, the convergence rate obeys a power law
with the exponent being determined by the smoothness of the
energy source and the norm used to measure the approximation error.

A finite amount of dissipation is a crucial ingredient in the entire approach,
that is, the weight $w$ has to satisfy $\|w \|_\infty <1$.
The simplest choice of a constant weight, $w(x,y)=\mu<1$, corresponds to
a dissipation which occurs at each collision at the boundary, for example,
an attenuation of the sound wave caused by an inelastic reflection at
the boundary of the cavity. Proper modelling of the damping parameters involved
is a crucial aspect of the method and is necessary
to describe realistic problems accurately~\cite{HarMorTanCha2019_tractor_yanmar}.
For example, a linear attenuation in the medium would
result in a path-length dependent weight $w(x,y)=
\exp(-2 \mu \cos(y))$.
This choice, however, does not obey the stipulated bound as orbits with
angles close to $y=\pm\pi/2$
have arbitrarily small path length, and hence small dissipation between
subsequent collisions. We could overcome this particular
problem by restricting
the angle of reflection to non-tangential collisions, that is,
$y \in (-(1-\epsilon)\pi/2, (1-\epsilon)\pi/2)$ for
a small $\epsilon > 0$, effectively constraining the permitted type of
energy source.
This however requires changing the Hilbert space and the projection
operators, as the validity of $c_k(D_y^mf) = (\mathrm{i}2k)^m c_k(f)$ and
$D_y P_k = P_k D_y$ is no longer given for a smooth function $f$ on an interval
instead of on a circle. One suitable choice could be
the space of functions in $H^s(\Omega_{\epsilon})$ with
vanishing weak derivatives $D^{\nu}f$ on the boundary. A suitable
basis is then the basis of Daubechies wavelets~\cite{Walker08}.

To illustrate the impact of Theorem \ref{thm:conv}, we perform
numerical simulations of circular billiards with constant damping
$w(x,y)=\mu$. As a proxy for the error estimate we use the distance
between approximations of subsequent order $\|f_{K+1} - f_K\|_{H^t}$,
which obeys essentially the same upper bound
\begin{equation}
  \|f_{K+1} - f_K\|_{H^t}
  \leq \|f - f_{K+1}\|_{H^t} + \|f - f_K\|_{H^t}
  \leq 2 C
    (1 + K^2)^{-\alpha/2} \|f_0\|_{H^s} \, .
    \label{eqn:prediction-numerics}
\end{equation}
Strictly speaking we have established this bound for integer vales of
$t_\ell$ and $s_\ell$ only. With a little more effort this could be
remedied by appealing to interpolation theory \cite{Triebel3}.
For simplicity of exposition we shall not pursue this here.
For our numerical considerations we take the liberty to
apply the bound above for non-integer values.
For the norm \(\|\cdot\|_{H^t}\), which estimates the truncation error,
we use the choices $t=(0,0)$, that is, the $L^2$ norm, and
$t=(1,1)$, a norm which is just outside the set of exponents guaranteeing
pointwise convergence.

The transfer operator's action on Fourier modes is given in
equation \eref{eq:TMatCirc}. In order to use it for a numerical test, we
have to use a representation for all Fourier modes, see
equation \eref{eq:T_matrix_circle}.
We show results for three different choices of the initial boundary
density $f_0$. They have in common that their support is given by
\[
  \mathrm{supp}(f_0) = \left\{
    (x, y): x\in \left[\pi/6,
      \pi/6 + 4\pi/3\right],
    y\in[-0.8, 1.2]\right\}.
\]
In order to define the boundary densities, we will use variables
scaled on this rectangle according to
$\tilde{x} = \left(x - \pi/6\right) / (4\pi/3)$ and
$\tilde{y} = (y + 0.8) / 2$ which take values between zero and one on
$\mathrm{supp}(f_0)$.
\begin{itemize}
\item Case $G$: a discontinuous function, that is, $f_0(x, y) = 1$ for
  \((x, y) \in \mathrm{supp}(f_0)\).
This function is contained in
$H^{(1/2-\epsilon,1/2-\epsilon)}(\Omega)$ for any small $\epsilon>0$.
For simplicity of exposition we will use, however, the value $s_2=1/2$
in the discussion of the numerical results below.

\item Case $W_1$: a continuous function given
  by $f_0(x, y) = \sqrt{\tilde{x}(1 - \tilde{x})}
\sqrt{\tilde{y}(1 - \tilde{y})}$ for
\((x, y) \in \mathrm{supp}(f_0)\).
 This function lies in
$H^{(1-\epsilon,1-\epsilon)}(\Omega)$ for any small $\epsilon>0$.
As before, we use the choice $s_2=1$ in the discussion below.

\item Case $W_2$: a smooth function given by
 $f_0(x, y) = \left(\sqrt{\tilde{x}(1 - \tilde{x})}
\sqrt{\tilde{y}(1 - \tilde{y})}\right)^3$ for
\((x, y) \in \mathrm{supp}(f_0)\).
This function
  lies in $H^{(2-\epsilon, 2-\epsilon)}(\Omega)$ for any small
  $\epsilon>0$ and we use the choice $s_2=2$ in our discussion.
\end{itemize}

\begin{figure}[h!]
  \centering
  \includegraphics[width=0.49\textwidth]{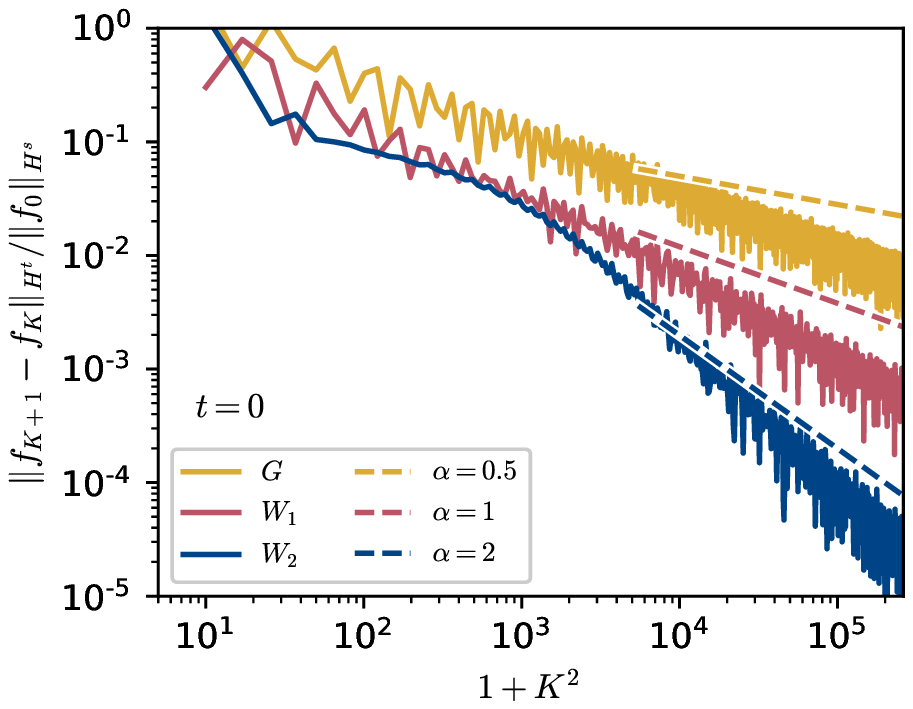}
  \includegraphics[width=0.49\textwidth]{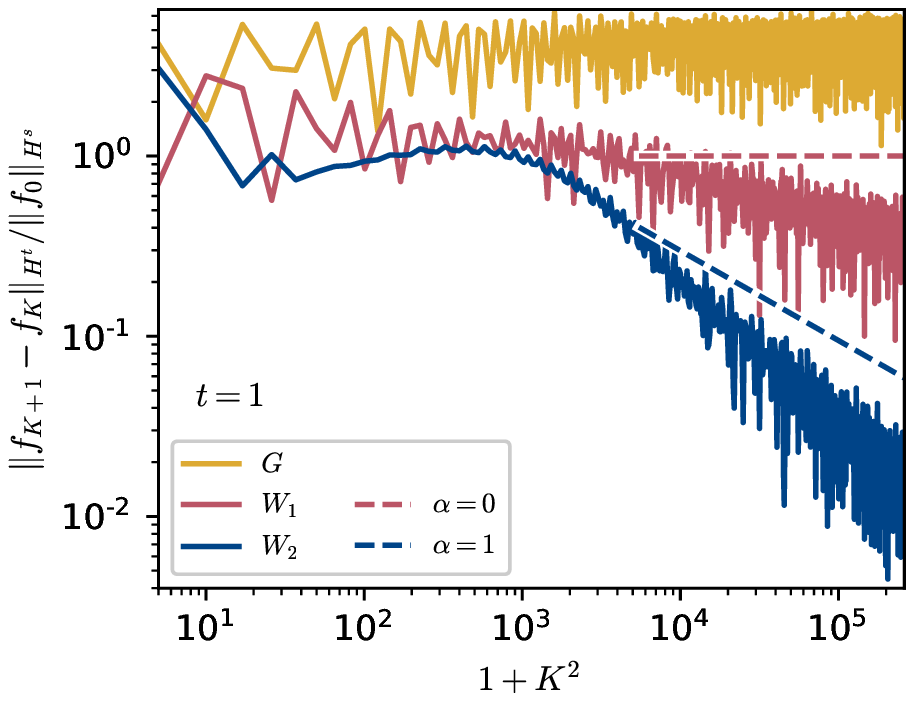}
  \caption{Error estimate $\|f_{K+1} - f_K\|_{H^t}$  for a circular billiard
    with constant damping $w(x,y)=\mu=0.9$ as a function of the truncation
    order $K$ on a double logarithmic scale. Left: $t_1=t_2 = 0$
    (convergence in $L^2$ norm), right: $t_1=t_2= 1$
    (point-wise convergence, in essence).
    Results are displayed for three
    different initial boundary densities $G$: $s_2=1/2$ (yellow, top),
    $W_1$: $s_2=1$ (red, middle),
    $W_2$: $s_2=2$ (dark blue, bottom), see text.
    Lines show the power law decay according to
    equation \eqref{eqn:prediction-numerics},
    $\alpha = s_2 - t_1$.
    \label{fig1}}
\end{figure}

The data shown in \Fref{fig1} confirm the upper bound in Theorem
\ref{thm:conv}.
For the $L^2$ norm, $t_1=t_2=0$, we observe,
in each case, convergence at a rate which is slightly faster than the
theoretical prediction $\alpha=s_2-t_1$. The power law decay of the truncation
error shows up for large values of $K$ and the onset of this scaling region
shifts towards larger values if the initial boundary density becomes smooth.
This should not come as a surprise, since the resolution of higher order
derivatives requires higher order Fourier modes.
For the parameter at the boundary of point-wise convergence $t=(1,1)$,  we see that the
discontinuous boundary density fails to
converge in line with our theoretical predictions.
While Theorem \ref{thm:conv} does not guarantee convergence in case
$W_1$ either, the numerical data suggest an extremely slow convergence
which is still consistent with the upper bound estimate $\alpha=s_2-t_1=1-1=0$. Finally,
for the smooth boundary density (case $W_2$) we observe a convergence rate
slightly faster than the theoretical prediction.

From a dynamical perspective, circular billiards are trivial since the
billiard map \eref{eq:T} is an integrable twist map.
In order to get an idea of how dynamical properties impact on convergence properties
we show numerical results for a deformed circle billiard which displays mixed regular and chaotic dynamics.
For the deformation
we choose the radius to depend on the polar angle $x$ according to
\begin{equation}
  \label{eqn:deformation}
  r(x) = 1 + \delta \cos(m x),
\end{equation}
where we choose $m = 3$ in the following.
Deformations of this kind are known in the
literature as  Lima\c{c}on billiards
\cite{BaeDul1997_Chaos_Limacon}.
We will cover the cases
$\delta = 0.01$ and $\delta = 0.1$. For larger values, the billiard
fails to be convex. In order to demonstrate the change in dynamical behaviour,
\Fref{fig:poincare_section} shows the Poincare plot of
the collision map $T$.
For a small value of the deformation, $\delta = 0.01$, one still observes
a fairly large number of invariant tori in accordance with general
KAM folklore. The larger perturbation shown in
\Fref{fig:poincare_section}, $\delta = 0.1$,
destroys most of the regular
motion and renders the system chaotic with a few exceptions, for example, the
highlighted period-$3$ island.

\begin{figure}[h!]
  \centering
  \includegraphics[width=0.93\textwidth]{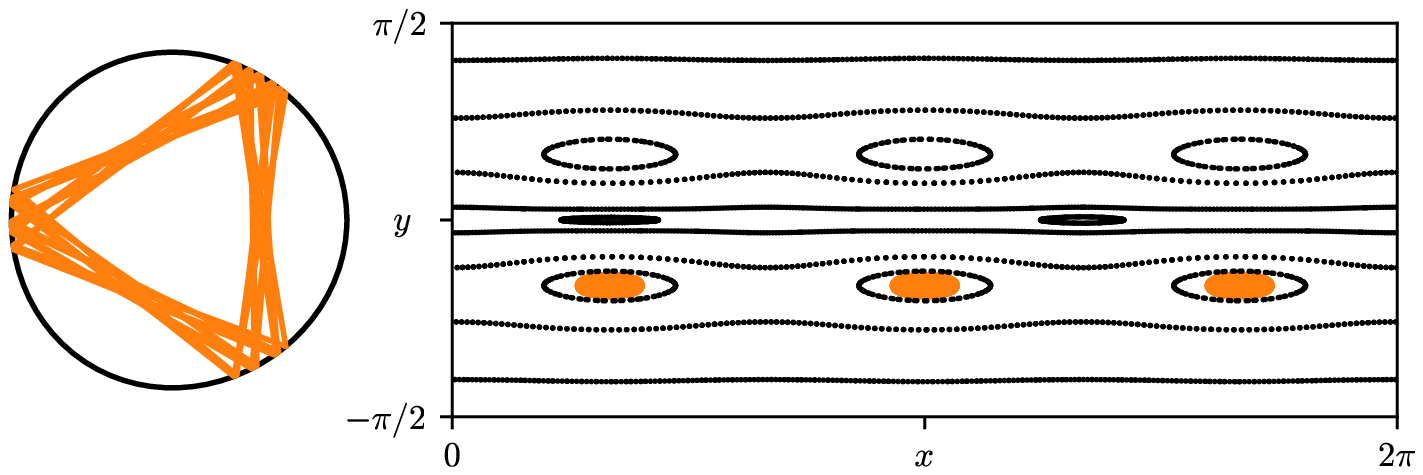}\\
  \includegraphics[width=0.93\textwidth]{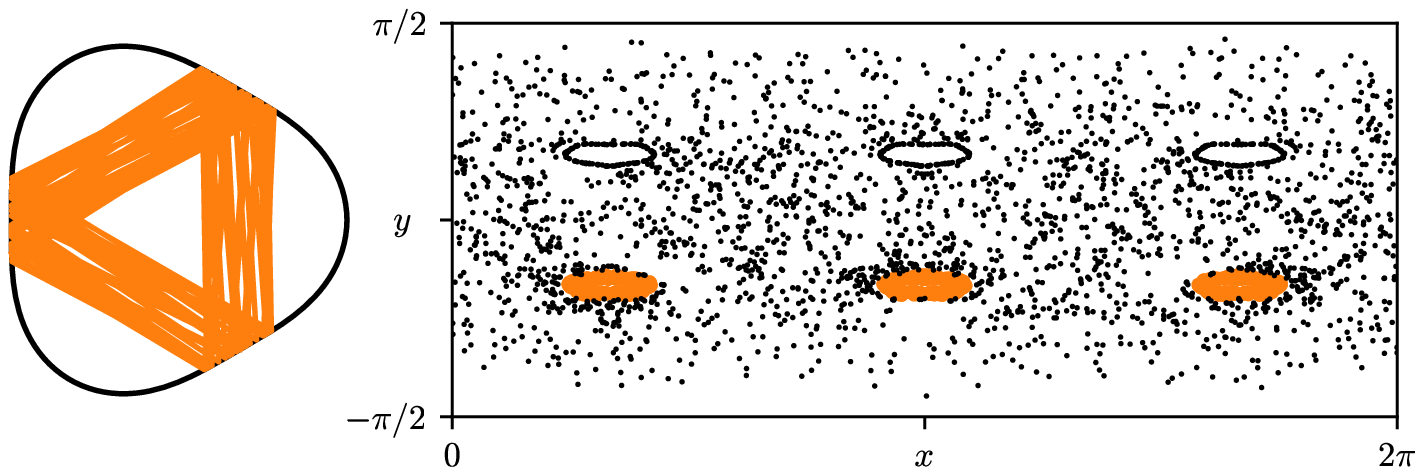}
  \caption{Billiard with orbit in configuration space
   (left) and Poincare plot of the boundary
    map $T$
    in the $(x, y)$ phase space (right)
    for a deformed billiard according to \eref{eqn:deformation}.
    Top: weak deformation of the circle ($m = 3, \delta=0.01$),
    bottom: strong but still convex deformation
    ($m = 3, \delta=0.1$).
    The orbit depicted in real space is highlighted in phase space
    as well.
    \label{fig:poincare_section}
  }
\end{figure}

In order to calculate the convergence of the energy
distribution we have to evaluate the matrix elements of the transfer
operator. For the circular billiard, the only non-zero entries take the value $\pm\mu$ and follow the structure given by equation \eqref{eq:T_matrix_circle}. Once the circle has been deformed, the analytic calculation of the matrix elements is no longer possible.
Even worse, the collision map is not given in
closed analytic form either, so that an efficient numerical calculation becomes
a nontrivial task (see the appendix for details).
However, we are able to reduce the calculation of the matrix elements to double
integrals with the kernel being given in closed analytic form,
see equation \eqref{eq:transfer_op_deformed}. Nevertheless,
the numerical evaluation is still time consuming, in particular,
since the matrix is no longer sparse.
Hence, we can only calculate finite approximations up to $K=30$.
In order to reach the scaling regime
(see \Fref{fig1} for comparison) we employ
a stronger damping of $\mu = 0.1$. The results for the error
measured in $L^2$ norm, that is, for the choice $t_1=t_2=0$, are shown in \Fref{fig2}.

\begin{figure}[h!]
  \centering
  \includegraphics[width=0.49\textwidth]{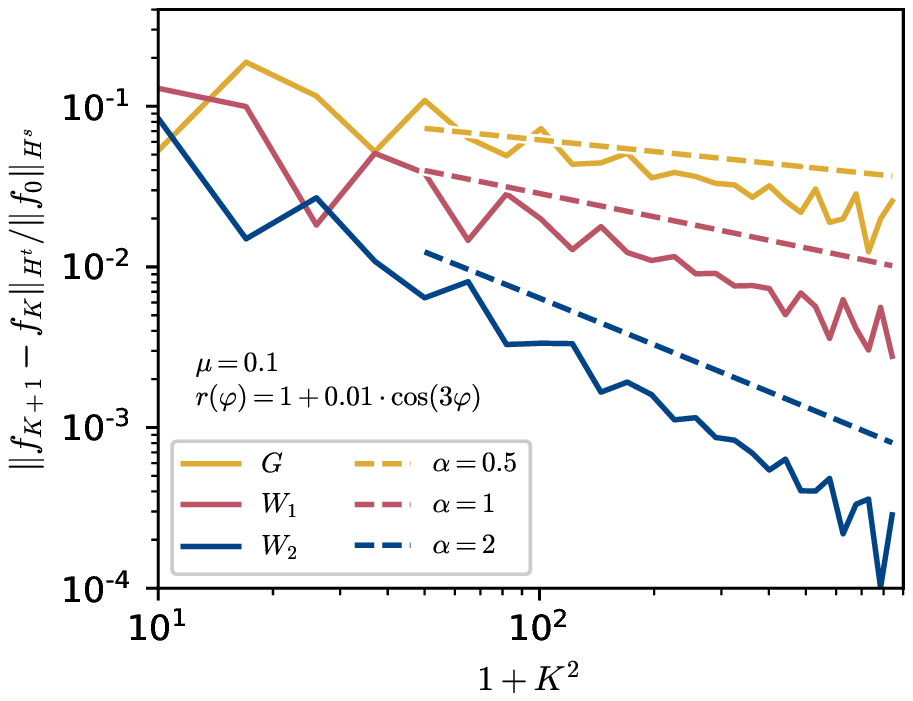}
  \includegraphics[width=0.49\textwidth]{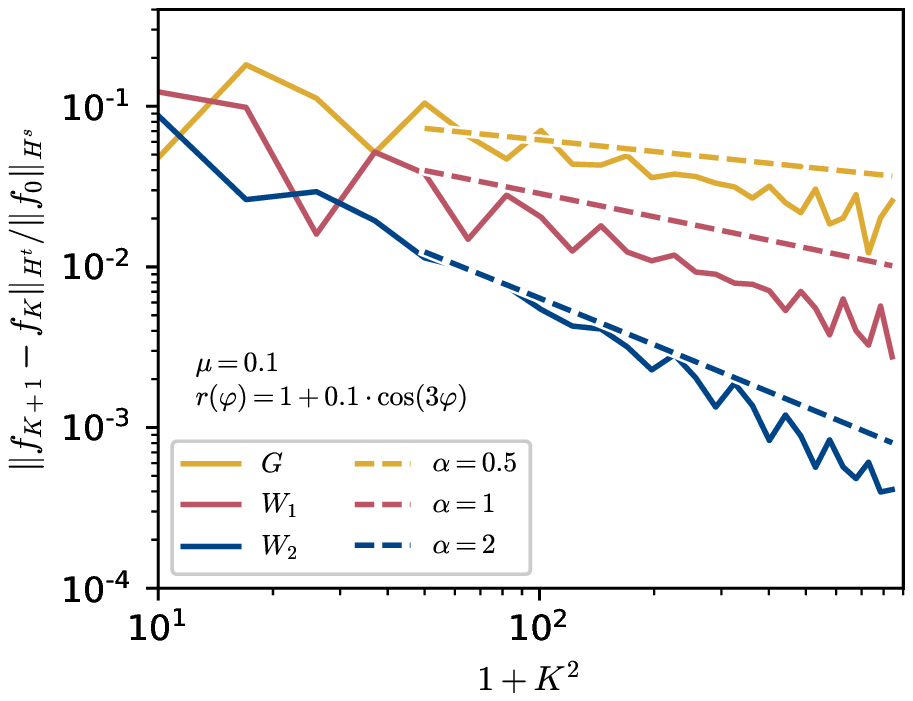}
  \caption{Error estimate $\|f_{K+1} - f_K\|_{H^t}$  in $L^2$ norm, $t=(0,0)$,
           for the energy density of a Lima\c{c}on billiard
           as a function of the truncation
           order $K$ on a double logarithmic scale.
           Constant damping $w(x,y)= \mu = 0.1$ and two deformations,
           $\delta = 0.01$ (left) and $\delta = 0.1$ (right), are considered.
           Results are displayed for the three
           different initial boundary densities $G$: $s_2=1/2$ (yellow, top),
           $W_1$: $s_2=1$ (red, middle), $W_2$: $s_2=2$ (dark blue, bottom),
           see \Fref{fig1}.
           Lines indicate a power law decay, $\alpha=s_2-t_1$,
           according to the rigorous estimate for circle billiards.
    \label{fig2}}
\end{figure}

It is quite remarkable that the decay of the error is apparently
almost unaffected by the degree of chaoticity. Hence the rigorous
error estimate of Theorem \ref{thm:conv} which covers the case
$\delta=0$ seems to have a wider range of applicability. While intuitively
such an observation would not be surprising for nearly
integrable cases it is quite counter-intuitive that the same
error estimate may hold as well in strongly chaotic situations.
However, our proof does not cover any of the deformed billiards
and there does not seem to be an obvious way how the methodology
can be generalised to these complicated cases. Nevertheless,
it is reaffirming that our study of a simple dynamical
system like the circular billiard has relevance for more complex dynamical behaviour.

\ack
The authors gratefully acknowledge the support of the research through
EPSRC grant EP/R012008/1.

\appendix

\section{Matrix elements}

Consider a convex billiard with boundary being given by $r(x)$ in polar
coordinates where $x$ denotes the polar angle (see, for example,
equation \eqref{eqn:deformation}). Denote by
$(x', y') = (T_x(x,y), T_y(x,y))$
the
collision map where $x$ and $x'$ label subsequent collisions with
the boundary. Using a standard representation in terms of Fourier basis
functions~\cite{Tann_JSV09}, the matrix elements $M_{l,k}$ of the transfer
operator read
\begin{eqnarray*}
  M_{l, k}
  &=& \frac{1}{2\pi^2}\int\limits_{0}^{2\pi}
      \int\limits_{-\pi/2}^{\pi/2}
      (\CC_{\phi}e_k)(x, y)\,
      \overline{e_l(x,y)}
      \, \mathrm{d}y \mathrm{d}x \nonumber
  \\
  &=&
      \frac{1}{2\pi^2}\int\limits_{0}^{2\pi}
      \!\!\int\limits_{-\pi/2}^{\pi/2}
      \!\!\mathrm{e}^{ \mathrm{i}k_1\phi_x(x, y)-\mathrm{i}l_1x +
      2\mathrm{i}k_2\phi_y(x, y) - 2\mathrm{i}l_2 y}
       \, \mathrm{d}y \mathrm{d}x
\end{eqnarray*}
with $k=(k_1,k_2)$ and $l=(l_1,l_2)$.

In case of the perfect circle we get a representation which is given
by a sparse matrix with only a few non-zero elements, close to the main
diagonal, namely
\begin{equation}
  (\CC_{\phi}e_l)(x, y) = \sum\limits_{k \in \Ztwo} M_{l, k}\cdot e_k(x, y)
\end{equation}
with the matrix elements
\begin{equation}
  M_{l, k} = (-1)^{k_1}\,\delta_{k_1, l_1}\,\delta_{l_2, k_1 + k_2}
  \,\quad\quad\,\, k, l \in \Ztwo.
  \label{eq:T_matrix_circle}
\end{equation}
This is the extension of equation \eref{eq:TMatCirc} to all Fourier modes
and it was used to calculate the values for \Fref{fig1}.

In order to eliminate the implicitly
defined collision map we change integration variables from $(x,y)$ to
$(x,x')$. Using $y_1(x,x')=y$ and $y_2(x,x')=y'$ for the two scattering angles
the matrix elements become
\begin{equation}\label{eq:transfer_op_deformed}
M_{l,k}=
      \frac{1}{2\pi^2} \int\limits_{0}^{2\pi}
      \!\!\int\limits_{0}^{2\pi}
      \left|\frac{\partial y_2(x,x')}{\partial x}\right|
      \mathrm{e}^{\mathrm{i}(k_1 x - l_1 x')}
      \mathrm{e}^{2\mathrm{i}(k_2 y_1(x, x') - l_2 y_2(x, x'))}
      \mathrm{d}x' \mathrm{d}x,
\end{equation}
where the additional factor is the Jacobian of the coordinate transformation.
In contrast to the collision map $T$, the
expressions $y_1(x,x')$ and $y_2(x,x')$ can be obtained in closed
analytic form so that equation \eqref{eq:transfer_op_deformed} is easier to
implement numerically.

\begin{figure}[h!]
  \centering
  \includegraphics[width=0.49\textwidth]{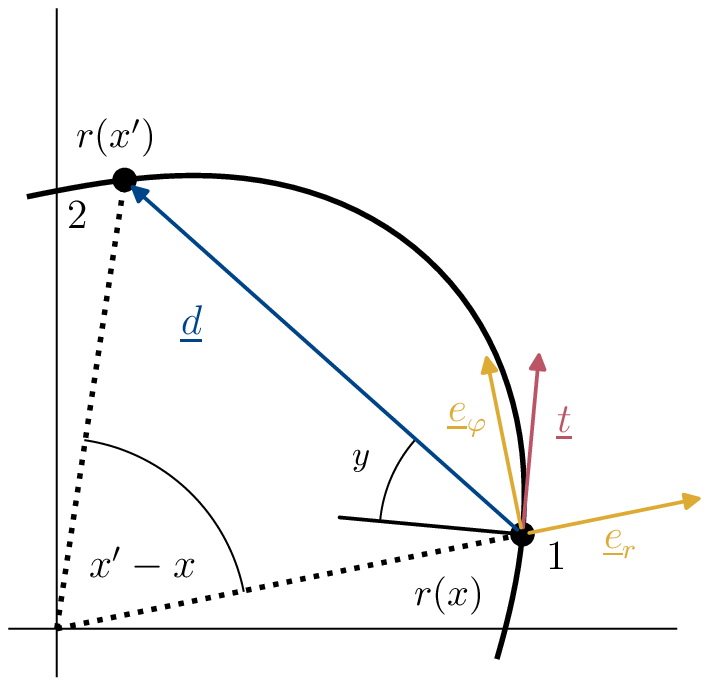}
  \caption{Geometric configuration of two subsequent collisions
in a convex billiard with a particle moving
from point 1 (with parameter value $x$) to point 2 (with parameter
value $x'$). We also depict the ray vector $\underline{d}$, the tangent vector $\underline{t}$,
and the unit vectors $\underline{e}_r$ and $\underline{e}_\varphi$ in
polar coordinates. \label{fig4}}
\end{figure}

\Fref{fig4} shows a sketch of two subsequent collisions.
The first scattering angle $y_1$ is given in terms of an
inner product
\begin{displaymath}
\sin(y_1)= \underline{d} \cdot \underline{t}/(|\underline{d}| |\underline{t}|)
\, .
\end{displaymath}
Since the position vector of the initial point is given by
$r(x) \underline{e}_r$ the tangent is easily obtained as
$\underline{t}=r'(x)\underline{e}_r+r(x) \underline{e}_\varphi$. The vector
separating the two points of collision is given in terms of the
local basis vectors by
\begin{displaymath}
\underline{d}= (r(x') \cos(x'-x)- r(x)) \underline{e}_r+ r(x') \sin(x'-x)
\underline{e}_\varphi \, .
\end{displaymath}
Hence the closed form expression for the first scattering angle reads
\begin{equation}\label{y1}
\sin(y_1)=\frac{r'(x)(r(x')\cos(x'-x)-r(x))+r(x) r(x') \sin(x'-x)}
{\sqrt{r^2(x)+(r'(x))^2} \sqrt{r^2(x)+r^2(x')-2 r(x) r(x') \cos(x'-x)}} \, .
\end{equation}
The second scattering angle is obtained by interchanging the two points
in \Fref{fig4}, i.e., by swapping $x$ and $x'$ in
equation \eqref{y1},
and including an additional minus sign for the outgoing angle
\[
\sin(y_2)=- \frac{r'(x')(r(x)\cos(x-x')-r(x'))+r(x') r(x) \sin(x-x')}
{\sqrt{r^2(x')+(r'(x'))^2} \sqrt{r^2(x')+r^2(x)-2 r(x') r(x) \cos(x-x')}} \, .
\]


\end{document}